\documentclass[final,5p,times,twocolumn]{elsarticle}
\usepackage{amsfonts}
\usepackage{algorithmic}
\usepackage{graphicx}
\usepackage{textcomp}
\usepackage{wrapfig}
\usepackage{comment,amsmath,amssymb,tabularx,epsfig,xspace,amsthm}
\usepackage[algoruled,vlined,linesnumbered]{algorithm2e}
\usepackage[dvipsnames]{xcolor}
\usepackage{anyfontsize}
\usepackage{t1enc}
\usepackage{caption}
\usepackage{subcaption}
\newtheorem{theorem}{Theorem}[section]
\newtheorem{corollary}{Corollary}

\newcommand{\ff}[1]{\textcolor{black}{#1}}

\SetKwInput{Input}{Input}
\SetKwInOut{Output}{Output}

\journal{Computer Networks}

\begin{document}

\begin{frontmatter}



\title{Optimistic Online Caching for Batched Requests}

\tnotetext[t1]{Research supported by Inria under the exploratory action MAMMALS. A preliminary version of this paper has been presented at IEEE ICC 2023~\cite{icc2023}.}


\author[inst1]{Francescomaria~Faticanti}

\author[inst1]{Giovanni~Neglia}

\affiliation[inst1]{organization={Inria},
            addressline={2004 route des Lucioles, BP 93}, 
            city={Sophia Antipolis},
            postcode={06902}, 
            country={France}}
            
\begin{abstract}
In this paper we study online caching problems where predictions of future requests, e.g., provided by a machine learning model, are available. Typical online optimistic policies are based on the Follow-The-Regularized-Leader algorithm and have higher computational cost than classic ones like LFU, LRU, as each update of the cache state 
requires to solve a constrained optimization problem. In this work we analysed the behaviour of two different optimistic policies in a \textit{batched} case, i.e., when the cache is updated less frequently in order to amortize the
update cost over time or over multiple requests. Experimental results show that such an optimistic batched approach
outperforms classical caching policies both on stationary and real traces.
\end{abstract}


\end{frontmatter}

\section{Introduction}\label{sec:intro}
Caching systems represent one of the most deeply studied research areas that span from the design of CPU hardware to the development of caching services in cloud computing, e.g., elastic caching systems for cloud and edge~\cite{carra2020elastic,carlsson2018worst}. The main objective of such systems is to reduce specific costs for the users, the network operator or the caching service provider. 
Caching policies have been studied under various assumptions on the arrival process of file requests. 
Recently online learning theory has been proposed to deal with caching settings where requests do not exhibit a regular pattern, and can be thought to be selected by an adversary~\cite{paschos2019learning,bhattacharjee2020fundamental,li2021online}.
Such an approach for the requests modeling stands in contrast to traditional stochastic models which can fail, e.g., in cases of small users' populations~\cite{leconte2016placing}.

Online caching has been studied in the online convex optimization (OCO) framework~\cite{shalev2012online} starting from the work~\cite{paschos2019learning}. In this setting, the main objective is to design algorithms that minimize the \textit{regret}, i.e., the difference between the cost incurred by the proposed solution and the cost of the optimal offline static solution that has complete knowledge of future requests over a fixed time horizon. Later contributions analyzed other online learning algorithms~\cite{salem2021no} and provided new lower bounds on the regret~\cite{bhattacharjee2020fundamental}.

Nowadays, thanks to the huge availability of data and resources in cloud systems, reliable predictions for future requests can be generated by machine learning (ML) models~\cite{gomez2015netflix,khanal2020systematic}.
Online algorithms that rely on such predictions are called \emph{optimistic}~\cite{mohri2016accelerating,mhaisen2022online}. 
References~\cite{mohri2016accelerating,rakhlin2013optimization} provide example of optimistic online algorithms
based on the Follow-The-Regularized-Leader (FTRL) and Online Mirror Descent (OMD) frameworks~\cite{shalev2012online}. 
Mhaisen et al.~\cite{mhaisen2022online} presented one of the first 
applications of optimistic online algorithms to a caching problem. 
They proved that  predictions, even if not perfectly accurate, can improve the performance of online algorithms. They designed an optimistic FTRL algorithm that operates on single requests requiring the cache to be updated 
each time a new file request is received. These updates are computationally very expensive, as they require to solve a constrained optimization problem, and can limit the applicability of online caching policies.
To amortize the update cost over time and over multiple requests a \textit{batched} approach can be adopted, where the caching system serves each request as it arrives, but updates the cache less frequently on the basis of the batch of requests collected since the last update~\cite{salem2021no}.
We stress that the batched approach does not cause any additional delay for the user. 

\ff{The novelty of this work resides in the study of
optimistic online caching policies able to work on 
batches of requests. Our main contributions are the following:\\
{\noindent \it 1)} We present a batched version of the optimistic caching policy in~\cite{mhaisen2022online}  and prove that it still enjoys sublinear regret.\\
{\noindent \it 2)} We introduce a new optimistic batched caching policy based on the per-component-based algorithm in~\cite{mohri2016accelerating}. \\
{\noindent \it 3)} We analytically characterize under which conditions each of these two caching policies outperforms the other.\\
{\noindent \it 4)} We determine when a batched operation provides better performance in terms of regret under different models for the predictions' error. \\
{\noindent \it 5)} We design optimistic versions of classical caching policies like LFU and LRU.\\
{\noindent \it 6)} We experimentally show, both on stationary traces and real ones, that our optimistic batched online caching policies outperform classical caching policies like LRU and LFU achieving both smaller service cost and per-request computational cost.}

\ff{The reminder of this paper is organized as follows. The next section discusses the main related works. Section~\ref{sec:model} introduces the system model and the problem description. In Section~\ref{sec:optimistic} we describe the optimistic caching framework and we present the main algorithms that take into account predictions: the one presented in~\cite{mhaisen2022online} and the one we propose. Section~\ref{sec:performance} presents an analysis of the regret bounds achieved by the two algorithms and a comparison between the single-request operation and the batched one. Experimental results are presented in Section~\ref{sec:experiments}. Finally, Section~\ref{sec:conclusions} concludes the paper.}

\section{Related Work}\label{sec:related}
\ff{Caching optimization problems have been deeply studied in the literature both on the offline and on the online perspective~\cite{paschos2020cache}. Several works have explored the offline static allocation of files under the assumption of knowing the requests~\cite{borst2010distributed,shanmugam2013femtocaching,poularakis2016exploiting}. On the online perspective, online caching policies based on gradient methods have been studied under the assumption of stochastic requests~\cite{ioannidis2010distributed,ioannidis2016adaptive}. In these works, the proposed algorithms have been evaluated under various performance metrics. We consider adversarial requests, i.e. the requests are thought as they are generated by an adversary trying to deteriorate the system's performance, and the regret as the main performance metric following the recent regret-based research on caching~\cite{bhattacharjee2020fundamental,li2021online,paria2021texttt,paschos2020online,salem2021no}. In this context, the main goal is to design algorithms with sublinear regret with respect to the time horizon leading to algorithms that behave as the optimal static solution in hindsight on average. Such online policies are called \textit{no-regret} algorithms~\cite{paschos2019learning}.}

\ff{Adversarial requests are considered in caching since Sleator and Tarjan's paper~\cite{sleator1985amortized} through the \textit{competitive ratio} metric. However, as proved in~\cite{andrew2013tale}, algorithms that ensure constant competitive ratio do not necessarily guarantee sublinear regret.} 

\ff{The main optimization framework adopted in this paper is the Online Convex Optimization (OCO). It was first introduced by Zinkevich~\cite{zinkevich2003online} showing that the projected gradient descent achieves sublinear regret bounds in the online setting. The works from Paschos et al.~\cite{paschos2019learning, paschos2020cache} were the first to apply the OCO framework to caching problems providing no-regret algorithms for the online caching problem. Bhattacharjee et al.~\cite{bhattacharjee2020fundamental} extended the work from Paschos et al. showing tighter lower bounds for the regret and proposing new online caching policies for the networked scenario based on the Follow-The-Perturbed-Leader (FTPL) algorithm. In our case, we consider the single-cache scenario and analyse the framework of the Follow-The-Regularized-Leader (FTRL) that has been proved one of the most promising algorithms for taking into account predictions in the online learning setting~\cite{mohri2016accelerating}. Indeed, as shown in~\cite{mhaisen2022optimistic}, the optimistic version of FTRL benefits more from the use of predictions with respect to the optimistic FTPL.}

\ff{The combination of predictions and caching has recently drawn attention given the significant usage of machine learning (ML) models for the computation of such predictions. The idea of exploiting predictions in the decision process has lead to the design of so called \textit{optimistic} online algorithms. Some works have already incorporated predictions in stochastic optimization~\cite{chen2018timely,huang2021online} assuming the requests and system perturbations to be stationary. In our work we do not make any assumption on the quality of the predictions that can be also thought as generated by and adversary. Mohri et al.~\cite{mohri2016accelerating} studied the regret performance of FTRL algorithms in adversarial settings including the predictions proving sublinear regret bounds. To the best of the authors' knowledge, Mhaisen et al.~\cite{mhaisen2022online} have been the first to apply optimistic online algorithms in the caching framework under adversarial settings. They proposed FTRL-based algorithms that, at each new request, update the cache state based on the previous incurred costs and the prediction for the next request. However, such algorithms imply the application of computationally-expensive operations, such as the projection on the domain set of the cache states~\cite{salem2023no,wang2015projection}, at each new request. To amortize the computational cost over time we propose to collect a \textit{batch} of requests before deciding the new cache state, leading to less frequent updates of the cache state. Theoretical analysis confirm that the size of such a batch does not affect the regret guarantees of the presented algorithms. A batched approach in caching has been presented in~\cite{salem2023no} but without taking into account predictions for future requests. Other optimistic online algorithms for caching are proposed in~\cite{mhaisen2022optimistic}. However, the proposed policies update the cache state at each new request, and the files are entirely stored in the cache, whilst, in line with recent works~\cite{paschos2019learning,mhaisen2022online}, we assume that the cache can store arbitrary fraction of files.}

\ff{The novelty of this work is in studying the performance of optimistic version of FTRL-based algorithms dealing with batches of requests. The account of batched requests reinforces also the use of predictions in the optimization process. It is reasonable indeed, when the predictions come from ML models, to involve a set of possible future requests in the predictions rather than a single future request. We show that the optimistic online batched algorithms introduced in this work present the best performance in terms of final miss-ratio and computational cost with respect the most practical and implemented caching policies.}

\section{System Description and Problem Formulation}\label{sec:model}
\subsection{System Model}
We consider the same system's setting  described in~\cite{salem2021no}. 
The system receives requests for equal-size files in the catalog $\mathcal{N} = \{1,2,\ldots,N\}$.
File requests are served by a single local cache  or by a remote server. In particular, a request for a file $i \in \mathcal{N}$ can be served by the cache for free or by a remote server incurring a per-file dependent cost $w_i \in \mathbb{R}^{+}$ (more details about our cost model below).
This cost can be related to the time needed to retrieve the file from a remote server, or be a monetary cost due to the utilisation of a third-party infrastructure for the file retrieval. We do not make any assumption on the requests arrival process, i.e., we analyse the system in an adversarial online setting where the requests can be thought as generated by an adversary trying to deteriorate system’s performance.\\
{\noindent \bf Cache State}. The local cache has finite capacity $k \in \{1,\ldots,N\}$, and it can store arbitrary fractions of files from the catalog as in~\cite{golrezaei2013femtocaching,paschos2019learning,mhaisen2022online}. We denote as $x_{t,i} \in [0,1]$ the fraction of file~$i$ stored in the cache at time $t$. The cache state, at time $t$, is then represented by the vector $\textbf{x}_t = [x_{t,i}]_{i\in \mathcal{N}}$ belonging to the set 
\begin{equation*}
\mathcal{X} = \left \{ x \in [0,1]^{N} | \sum_{i\in \mathcal{N}} x_{i} = k \right \}.
\end{equation*}
The set $\mathcal{X}$ is the capped simplex defined by the capacity constraint of the local cache.
It is sometimes convenient to express the cache capacity as a fraction of the catalog size, i.e., $k = \alpha N$, where $\alpha \in [0,1]$. 

{\noindent \bf Cache Updates}.
Caching decisions are taken after batches  (potentially of different sizes) of requests have been served. Formally, at each time-slot $t=1,\ldots,T$ the system collects $R_t$ requests from the users and then it may updates the cache state. The request process can then be represented as a sequence of vectors $\mathbf{r}_t = (r_{t,i} \in \mathbb{N} : i \in \mathcal{N})$, $\forall t$, where $r_{t,i}$ denotes the number of requests for file~$i$ in the $t$-th timeslot. The request process belongs then to the set
\begin{equation*}\label{eq:requests_set}
\mathcal{R} = \left \{ \mathbf{r}_t \in \mathbb{N}^N, t = 1, \dots, T | \sum_{i \in \mathcal{N}} r_{t,i} = R_t \right \}.
\end{equation*}
For some results we will rely on the following additional assumption (already proposed in~\cite{salem2021no}): 

{\it Assumption 1}. Every batch contains the same number of requests (i.e., $R_t=R$ for all $t\in {1, \dots, T}$) and the number of requests for each file within the batch is bounded by $h$ (i.e., $r_{n}^t \in \{0,\ldots,h\}$).

{\noindent \bf Cost Function}. For each new  batch of requests  $\textbf{r}_t$ the system pays a cost proportional to the missing fraction $(1-x_{t,i})$ for each file $i \in \mathcal{N}$ from the local cache. More formally: 
\begin{equation}\label{eq:cost}
f_{\mathbf{r}_t}(\mathbf{x}_t) = \sum_{i=1}^N w_i r_{t,i} (1-x_{t,i}).
\end{equation}
The sum is weighted by the cost $w_i$ and by the number of times $r_{t,i}$ file $i$ is requested in the batch $\mathbf{r}_t$.

{\noindent \bf Predictions}. Predictions for the next batch of requests can be the output of a ML model such as a neural network. Such prediction models can be similar to those used in streaming services like Netflix to provide recommendations to users on the basis of their history view~\cite{gomez2015netflix}. We assume that the predictor provides an estimate for the number of requests for each file in the next time-slot. We indicate with $\tilde{r}_{t+1,i}$ the prediction of the number of requests for file $i$ at time $t+1$.  It is then  possible to directly estimate the  gradient of the cost function in that time-slot. More formally, we denote by $\tilde{\mathbf{g}}_{t+1}$ the prediction of $\mathbf{g}_{t+1} = \nabla f_{\mathbf{r}_{t+1}}(\mathbf{x}_{t+1})$, the gradient of the cost function at time $t+1$, where $\tilde{g}_{t+1,i} = -w_i\tilde{r}_{t+1,i}$.

\subsection{Online Caching Problem}
We can fit our caching problem in the \textit{Online Convex Optimization} (OCO) framework~\cite{zinkevich2003online,hazan2016introduction}, where a learner (in our case the caching system) has to take a decision $\mathbf{x}_t$ from a convex set $\mathcal{X}$ at each time slot $t$ before the adversary selects the cost function $f_{\mathbf{r}_t}$
, i.e., the learner changes the cache state before experiencing the cost. Hence, the main objective is to devise a caching policy $\mathcal{A}$ that, at each time-slot $t$, computes the cache state $\mathbf{x}_{t+1}$ for the next time-slot given the current cache state $\mathbf{x}_t$, the whole history up to time $t$ ($(\mathbf{x}_1,r_1),\ldots,(\mathbf{x}_t,r_t)$), and possibly the predictions for the next time-slot.
As it is common in online learning, the main performance metric for the caching policy $\mathcal{A}$ is
the regret defined as
\begin{equation}\label{eq:regret}
    R_{T}(\mathcal{A}) = \sup_{\{\mathbf{r}_1,\ldots,\mathbf{r}_T\}} \left \{ \sum_{t=1}^T f_{\mathbf{r}_t}(\mathbf{x}_t) - \sum_{t=1}^T f_{\mathbf{r}_t}(\mathbf{x}^\star)\right \}.
\end{equation}
This function denotes the difference between the total cost obtained by the online policy $\mathcal{A}$ over a time horizon $T$, and the total cost of the best caching state $\mathbf{x}^\star$ in hindsight , i.e., $\mathbf{x}^\star = \arg \min_{x \in \mathcal{X}} \sum_{t=1}^T f_{\mathbf{r}_t}(\mathbf{x})$. The supremum in~\eqref{eq:regret} indicates an adversarial setting for the regret definition, i.e., the regret is measured against an adversary that generates requests trying to deteriorate the performance of the caching system. The main goal in this setting is to design a caching policy $\mathcal{A}$ that achieves sublinear regret, $R_T(\mathcal{A}) = o(T)$. This ensures a zero average regret as $T$ grows implying that the designed policy behaves on average as the optimal static one.

In what follows, given a sequence of vectors $(\mathbf{y}_1, \mathbf{y}_2, \dots, \mathbf{y}_t, \dots)$, we denote their aggregate sum up to time $t$ as $\mathbf{y}_{1:t}\triangleq \sum_{s=1}^t \mathbf{y}_s$.

\section{Optimistic Caching}\label{sec:optimistic}
As highlighted in~\cite{mhaisen2022online}, an optimistic caching policy can exploit, at each time-slot $t$, predictions for the requests at time $t+1$ in order to compute the caching state $\mathbf{x}_{t+1}$. 
The general scheme for optimistic online caching is described in Algorithm~\ref{algo:optCaching}. Given an initial feasible solution $\mathbf{x}_1 \in \mathcal{X}$, the cache operates at each time-slot $t$ as follows: i) the new batch of requests $\mathbf{r}_t$ is revealed; ii) based on the current cache state $\mathbf{x}_t$, the cache incurs the cost $f_{\mathbf{r}_t}(\mathbf{x}_t)$; iii) the cache receives the prediction $\tilde{\mathbf{g}}_{t+1}$ for the next time-slot, and iv)~based on such predictions and on all the history up to time $t$ ($(\mathbf{x}_1,r_1),\ldots,(\mathbf{x}_t,\mathbf{r}_t)$), it computes the next cache state $\mathbf{x}_{t+1}$.

In the OCO literature,  algorithms exploiting predictions are usually variants of the \textit{Follow-The-Regularized-Leader} (FTRL) algorithm~\cite{mcmahan2017survey,mohri2016accelerating}. 
\ff{The classic \textit{Follow-The-Leader} (FTL) algorithm~\cite{littlestone1994weighted} greedily selects the next state in order to minimize the aggregate cost over the past, i.e.,
\begin{equation*}
\mathbf{x}_{t+1} := \arg\min_{\mathbf x \in \mathcal X} \sum_{s=1}^t f_{\mathbf r_s} (\mathbf x)= \arg\min_{\mathbf x \in \mathcal X} \mathbf g_{1:t}^\intercal \mathbf x, 
\end{equation*}
where the last equality follows from the linearity of the cost functions.} The linearity of the problem leads FTL to commit to store entirely some files (i.e., $\mathbf x_{t+1} \in \{0,1\}^N$), but this can be exploited by the adversary and leads to a linear regret. The FTRL algorithm improves the performance of FTL by adding a non-linear proximal regularization term, which leads to more cautious updates.\footnote{
A regularizer 
is proximal if $\arg \min_{\mathbf{x} \in \mathcal{X}} r_{t}(\mathbf{x}) = \mathbf{x}_t$.
} Let $r_t(\mathbf x)$ be the regularization function used at time~$t$ (to be specified later). The FTRL algorithm's update step is given by
\begin{equation}\label{eq:update}
    \mathbf{x}_{t+1} := \arg \min_{\mathbf{x} \in \mathcal{X}} \left \{ r_{1:t}(\mathbf{x}) + (\mathbf{g}_{1:t}+\tilde{\mathbf{g}}_{t+1})^{\top}\mathbf{x}\right \}.
\end{equation}
As we are going to see, the function to minimize in~\eqref{eq:update} is a quadratic function. The Problem~\ref{eq:update} can then be solved through popular solvers like \texttt{CVX}, but the presence of the constraint $\mathbf x \in \mathcal X$  makes the update a potentially expensive operation, motivating the batched operation we propose.

\begin{algorithm}[t]
\Input{$\mathcal{N},k,x_1\in \mathcal{X}$}
\For{$t = 1,\ldots,T$}{
    Receive the batch of requests $\mathbf{r}_t$\;
    Incur cost $f_{\mathbf{r}_t}(\mathbf{x}_t)$\;
    Receive the new prediction $\tilde{\mathbf{g}}_{t+1}$\;
    Compute $\mathbf{x}_{t+1}$ taking into account $\tilde{\mathbf{g}}_{t+1}$ and the history $((\mathbf{x}_1,\mathbf{r}_1),\ldots,(\mathbf{x}_t,\mathbf{r}_t))$ according to~\eqref{eq:update}.
}
\caption{\texttt{Optimistic Online Caching}}
\label{algo:optCaching}
\end{algorithm}
In what follows, we describe two particular FTRL instances applied to our caching problem. The two instances  
differ by the specific regularization function used  in~\eqref{eq:update} for updating the cache state (line 5 of Algorithm~\ref{algo:optCaching}). 

\subsection{Optimistic Bipartite Caching (OBC)}\label{sec:OBC}

The first algorithm is called \textit{Optimistic Bipartite Caching} (OBC) and was introduced in~\cite{mhaisen2022online} for a  bipartite caching system with a single request at each time-slot. OBC adopts as proximal regularizer
\begin{equation}\label{eq:regMhaisen}
 r_t(\mathbf{x}) = \frac{\sigma_t}{2} \lVert \mathbf{x} - \mathbf{x}_t \rVert^2, t \ge 1,
\end{equation}
with the following parameters
\begin{equation}\label{eq:paramMhaisen}
\sigma_t = \sigma(\sqrt{h_{1:t}}-\sqrt{h_{1:t-1}}), \quad \text{where} \quad h_t = \lVert \mathbf{g}_t - \tilde{\mathbf{g}}_t \rVert^2,
\end{equation}
and $\sigma \ge 0$. 
The regularizer $r_{1:t}(\mathbf x)$ is 1-strongly convex with respect to the norm $\lVert \mathbf{x}\rVert_{(t)} = \sqrt{\sigma_{1:t}}\Vert \mathbf{x}\rVert$ whose dual norm we denote by $\lVert \mathbf{x}\rVert_{(t),\star}$.
The regularizer depends  on the Euclidean distance between the actual gradient $\mathbf g_t$ and the predicted one~$\tilde{\mathbf g}_t$. Qualitatively, if  predictions are very accurate, $r_{1:t}(\mathbf x)$ is small  and then the update in~\eqref{eq:update} will focus on minimizing the (predicted) aggregate cost $(\mathbf{g}_{1:t}+ \tilde{\mathbf{g}}_{t+1})^\intercal \mathbf x$. On the contrary, if predictions are not accurate, the regularizer will lead to more cautious updates.
The regularization function can then be interpreted as an implicit adaptive learning rate~\cite{mohri2016accelerating}: as gradient predictions become more accurate the algorithm \textit{accelerates} towards the minimum of the aggregate cost $(\mathbf{g}_{1:t}+ \tilde{\mathbf{g}}_{t+1})^\intercal \mathbf x$.

In the next section, we present theoretical guarantees on the OBC's regret for the batched setting considered in this paper.

\subsection{Per-Coordinate Optimistic Caching (PCOC)}\label{sec:PCOC}

Mohri et al.~\cite[Corollary~2]{mohri2016accelerating} proposed an FTRL algorithm where the regularization function decomposes over the coordinates and thus the acceleration occurs on a per-coordinate basis. In this case, if gradient predictions are more accurate on certain coordinates, the algorithm will accelerate the convergence of such coordinates. Here we present a generalization of this algorithm, called \textit{Per-Coordinate Optimistic Caching} (PCOC), which introduces a generic parameter $\sigma$ in the definition of the regularization function:
\begin{equation}\label{eq:regulNew}
r_{t}(\mathbf{x}) = \sum_{i=1}^N \sum_{s=1}^t  \frac{\sigma_{t,i}}{2}(x_i-x_{s,i})^2,
\end{equation}
where $\sigma_{t,i} = \sigma (\Delta_{t,i} - \Delta_{t-1,i})$, and $\Delta_{s,i} = \sqrt{\sum_{a=1}^s (g_{a,i}-\tilde{g}_{a,i})^2}$. The function $r_{0:t}(\mathbf x)$   is 1-strongly convex with respect to\footnote{
    With some abuse of notation we use the same symbols  (resp. $\lVert \cdot \rVert_{(t)}$ and $\Vert \cdot \rVert_{(t),*}$) to denote the norms and the dual norms for OBC and PCOC. The interpretation of the symbols should be clear from the context.
}
\begin{equation}\label{eq:normNew}
\lVert \mathbf{x} \rVert_{(t)}^2 = \sum_{i=1}^N \sigma_{1:t,i} x_i^2 , \quad \text{with} \quad \Vert \mathbf{x} \rVert_{(t),*}^2 = \sum_{i=1}^N\frac{x_i^2}{\sigma_{1:t,i}}.
\end{equation}

\section{Performance Analysis}\label{sec:performance}
Here we prove theoretical guarantees for the regret bounds of the algorithms presented in the previous section
in the case of a single cache and multiple requests at each time-slot. 
We show that both algorithms enjoy sublinear regrets even if gradient predictions are inaccurate.
\subsection{Regret bound of OBC with single cache and $R$ requests}
We extend the regret bound in~\cite[Theorem~1]{mhaisen2022online} to the case of batched requests, but we also improve the coefficients taking into account the capacity constraint.
\begin{theorem}\label{thm:OBC}
The regret of OBC is bounded as follows:
\begin{equation}\label{eq:thm1}
\centering
R_T(OBC)\le 2\sqrt{2\min\{k,N-k\}\cdot \sum_{t=1}^T {\lVert \mathbf{g}_t - \tilde{\mathbf{g}}_t \rVert}^2}.
\end{equation}
\end{theorem}
\begin{proof}
We start from the inequality in~\cite[Theorem~1]{mohri2016accelerating},
\begin{equation}
R_T \le r_{1:T}(\mathbf{x}^\star) + \sum_{t=1}^T {\lVert \mathbf{g}_t - \tilde{\mathbf{g}}_t \rVert}_{(t),\star}^2, \quad \forall \mathbf{x}^\star \in \mathcal{X}.
\end{equation}
Substituting the regularization functions we obtain
\begin{equation}\label{eq:starting}
R_T \le \frac{\sigma}{2} \sum_{t=1}^T (\sqrt{h_{1:t}} - \sqrt{h_{1:t-1}}) \lVert \mathbf{x}^\star - \mathbf{x}_t \rVert^2 + \sum_{t=1}^T \frac{h_t}{\sigma \sqrt{h_{1:t}}}
\end{equation}
In our case, as highlighted in~\cite{paschos2019learning}, the Euclidean diameter of $\mathcal{X}$ is upper bounded by $\Delta$
\begin{equation}
\lVert \mathbf{x} - \mathbf{x}_t \rVert^2 \le \Delta^2 \triangleq \min\{2k, 2(N-k)\}, \forall \mathbf x, \mathbf x_t \in \mathcal X.
\end{equation}
Introducing $\Delta$ in~\eqref{eq:starting}, and using \cite[Lemma~3.5]{auer2002adaptive} it follows
\begin{equation}
\begin{split}
R_T \le \frac{\sigma}{2} \Delta^2 \sum_{t=1}^T (\sqrt{h_{1:t}} - \sqrt{h_{1:t-1}}) + \sum_{t=1}^T \frac{h_t}{\sigma \sqrt{h_{1:t}}}\\
\le \frac{\sigma}{4} \Delta^2 \sqrt{h_{1:T}} + \frac{2}{\sigma} \sqrt{h_{1:T}}
= (\frac{\sigma}{2} \Delta^2 + \frac{2}{\sigma}) \sqrt{h_{1:T}}.
\end{split}
\end{equation}
Setting $\sigma = 2/\Delta$ we obtain the desired bound.
\end{proof}
Theorem~\ref{thm:OBC} shows that the regret bound depends on the cache size, and on the accuracy in the predictions. The algorithm enjoys a zero regret if the cache is able to store the complete catalog, i.e., $k=N$, or if predictions are perfect, i.e., $\tilde{g}_t = g_t$. On the other hand, even if predictions are imperfect, OBC may guarantee sublinear regret, as shown by the following corollary.
\begin{corollary}
Under \textit{Assumption 1},
\begin{equation}
    \centering
    R_T \le 2\lVert w\rVert_{\infty}\sqrt{2\min\{k,N-k\}TRh} = O(\sqrt{T}).
\end{equation}
\end{corollary}
The proof easily follows from $\lVert \tilde{\mathbf{g}}_t - \mathbf{g}_t \rVert^2 \le \lVert w\rVert_{\infty}^2Rh$ under \textit{Assumption 1}.

\subsection{Regret bound of PCOC}
The following proof follows the steps in~\cite[Corollary 2]{mohri2016accelerating}, introducing the adjustable parameter $\sigma \ge 0$ in the definition of the regularizer~\ref{eq:regulNew} and taking into account that $x_i \in [0,1]$ for our caching application.
\begin{theorem}\label{thm:PCOC}
The regret of PCOC is bounded as follows
\begin{equation}\label{eq:finalmohriB}
R_T(PCOC) \le 2\sum_{i=1}^N \sqrt{\sum_{t=1}^T (g_{t,i}-\tilde{g}_{t,i})^2}.
\end{equation}
\end{theorem}
\begin{proof}
From~\cite[Theorem~3]{mohri2016accelerating}, applying the regularization function defined in~\eqref{eq:regulNew} and the norms defined in~\eqref{eq:normNew}, we obtain
\begin{gather}
    R_T \le \frac{\sigma}{2}\sum_{i=1}^N\sum_{s=1}^T (\Delta_{s,i}-\Delta_{s-1,i})(x_i-x_{s,i})^2 + \sum_{t=1}^T \lVert \mathbf{g}_t - \tilde{\mathbf{g}}\rVert_{(t),\star}^2 \nonumber \\
    \stackrel{\text{\tiny (a)}}{\le} \frac{\sigma}{2}\sum_{i=1}^N\sqrt{\sum_{t=1}^T (g_{t,i}-\tilde{g}_{t,i})^2} + \sum_{t=1}^T \lVert \mathbf{g}_t - \tilde{\mathbf{g}}\rVert_{(t),\star}^2  \nonumber \\ \stackrel{\text{\tiny (b)}}{\le}\frac{\sigma}{2}\sum_{i=1}^N\sqrt{\sum_{t=1}^T (g_{t,i}-\tilde{g}_{t,i})^2} + \frac{2}{\sigma}\sum_{i=1}^N\sqrt{\sum_{t=1}^T (g_{t,i}-\tilde{g}_{t,i})^2}  \nonumber \\
= \left (\frac{\sigma}{2} + \frac{2}{\sigma} \right )\sum_{i=1}^N\sqrt{\sum_{t=1}^T (g_{t,i}-\tilde{g}_{t,i})^2},
\end{gather}
where (a) follows from $(x_i-x_{s,i})^2 \le 1$ and the results of the telescopic sum $\sum_{s=1}^t \Delta_{s,i}-\Delta_{s-1,i}$, and (b) from the application of~\cite[Lemma 3.5]{auer2002adaptive} to $\sum_{t=1}^T \lVert \mathbf{g}_t - \tilde{\mathbf{g}}\rVert_{(t),\star}^2$ once the definition of dual norm in~\eqref{eq:normNew} has been applied. For the minimization of the regret bound we can set $\sigma = 2$. 
\end{proof}
Similar to OBC, PCOC has zero regret under perfect predictions, and sublinear regret under \textit{Assumption 1}.
\begin{corollary}
Under \textit{Assumption 1},
\begin{equation}
    \centering
     R_T \le 2Nh\lVert w\rVert \sqrt{T} = O(\sqrt{T}).
\end{equation}
\end{corollary}
The proof follows from $(g_{t,i}-\tilde{g}_{t,i})^2 \le \lVert w\rVert^2h^2$ under \textit{Assumption 1}.

\subsection{Comparison between the two regret bounds}
We compare the two bounds presented above in two specific scenarios for the prediction error: i)~a constant error on each component of the gradient, and ii)~a prediction error proportional to the popularity of the files in the catalog.

In the first case, OBC presents a better bound with respect to the one obtained by PCOC. In fact, say that $|g_{t,i} - \tilde g_{t,i}| = \epsilon$ for each $i $ and $t$, then $R_T(OBC) = 2\sqrt{2\min\{k,N-k\}NT\epsilon^2} \le  2N\sqrt{T\epsilon^2} = R_T(PCOC)$.

In the second case, PCOC may perform better because it specifically takes into account the heterogeneity of the prediction error across the components. We deviate here from the adversarial request model and consider 
that 1) requests arrive according to a Poisson process with rate~$\lambda$, and 2) a request is for file $i$ with probability $p_i$ independently from the past~\cite{irm-fagin-1977}. Moreover, we assume 
the algorithm is executed every time unit, and per-file costs equal $1$. In this case,  $g_{t,i} \sim \textrm{Poisson}(\lambda p_i)$ for each $i\in \mathcal{N}$.
We compute the expected value of the bounds in \eqref{eq:thm1} and in \eqref{eq:finalmohriB}, assuming that the cache can store a fraction $\alpha$ of the catalog ($k= \alpha N$), and $\tilde g_{t,i} = \lambda p_i$, i.e., we have a perfect predictors for the expected number of future requests. For the OBC bound, we obtain
\begin{gather}
    \mathbb{E}\left [2\sqrt{2\alpha N \sum_{t=1}^T \sum_{i=1}^N (g_{t,i}-\tilde{g}_{t,i})^2} \right ] \le \nonumber \\
    \le 2\sqrt{2\alpha N \sum_{t=1}^T \sum_{i=1}^N \mathbb{E}[ (g_{t,i}-\tilde{g}_{t,i})^2]} = \nonumber \\
    = 2\sqrt{2\alpha N \sum_{t=1}^T \sum_{i=1}^N \lambda p_i} = 
    2\sqrt{2\alpha \lambda N T}\quad.\label{eq:OBCexp}
\end{gather}
For the PCOC bound, we obtain
\begin{gather}
    \mathbb{E}\left [\sum_{i=1}^N\sqrt{\sum_{t=1}^T (g_{t,i}-\tilde{g}_{t,i})^2} \right ] \le
    \sum_{i=1}^N\sqrt{\sum_{t=1}^T \mathbb{E}[(g_{t,i}-\tilde{g}_{t,i})^2]} \nonumber \\
    = 2\sum_{i=1}^N\sqrt{\sum_{t=1}^T \lambda p_i} = 2\sum_{i=1}^N\sqrt{ T  \lambda p_i}\quad.\label{eq:PCOCexp}
\end{gather}
Comparing the two bounds~\eqref{eq:PCOCexp} and~\eqref{eq:OBCexp}, 
we find that~\eqref{eq:PCOCexp} is a smaller than~\eqref{eq:OBCexp} when $\alpha \ge \left ( \sum_{i=1}^N\sqrt{p_i}\right )^2/(2N\sum_{i=1}^N p_i$). If $p_i$ obeys to a Zipf law with exponent $\beta$, we can numerically find from the inequality the minimum value of $\alpha$ such that the bound of~\eqref{eq:PCOCexp} is tighter. In Figure~\ref{fig:aVSb} we can notice that the threshold for $\alpha$ decreases as $\beta$ increases. In the case of a uniform popularity distribution ($\beta = 0$), OBC outperforms PCOC unless the cache can store at least  half of the catalog. As the popularity distribution becomes more skewed, PCOC is expected to perform better than OBC in terms of regret bound, but for very small caches. 
\begin{figure}
    \centering
    \includegraphics[width=0.8\linewidth]{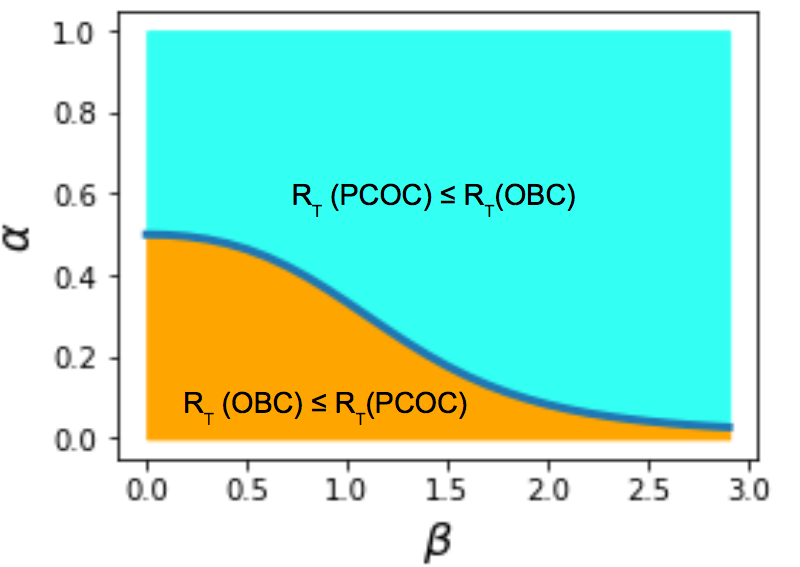}
    \caption{OBC vs PCOC, different regimes for the regret as a function of the Zipf exponent ($\beta$) and the relative cache size ($k = \alpha N$).}
    \label{fig:aVSb}
\end{figure}
\subsection{Batch Selection}\label{sec:SinglevsBatch}
We maintain the Poisson assumption about the request arrival process and evaluate what is the effect of requests batching on the regret, focusing on the bound in Theorem~\ref{thm:OBC} (the same analysis can be carried out on the bound in Theorem~\ref{thm:PCOC}). 
We analyse the expected value of such bound in a general batched-requests setting where the caching decisions are taken every $\tau$ among an overall time interval of $\Theta$ time units where a single request is available at each time. Looking at the expected value of the regret bound we have:
\begin{equation}
    \mathbb{E} \left[R_{\Theta/\tau}\right] \le \mathbb{E}\left[C\sqrt{\sum_{t =1}^{\Theta/\tau}\sum_{i=1}^N (g_{t,i}-\tilde{g}_{t,i})^2}\right],
\end{equation}
where $C \triangleq 2\sqrt{2\min \{k,N-k\}}$. In this case we have $g_{t,i}\sim \textrm{Poisson}(\lambda_i\tau)$. For the predictions $\tilde{g}_{t,i}$ we consider two options: i) they  coincide with the expected number of future requests, or ii) they coincide with the requests seen during the previous times-lots.

In the first case we have
\begin{small}
\begin{align}\label{eq:singleVSbatch_1}
\mathbb{E}\left[C\sqrt{\sum_{t =1}^{\Theta/\tau}\sum_{i=1}^N (g_{t,i}-\tilde{g}_{t,i})^2}\right] \stackrel{\text{\tiny (a)}}\le 
C\sqrt{\sum_{t =1}^{\Theta/\tau}\sum_{i=1}^N\mathbb{E}\left[(g_{t,i}-\tilde{g}_{t,i})^2\right]} = \notag \\
= C\sqrt{\sum_{t =1}^{\Theta/\tau}\sum_{i=1}^N \textrm{Var}(g_{t,i})} = C\sqrt{\sum_{t =1}^{\Theta/\tau}\sum_{i=1}^N \lambda_i\tau} = C\sqrt{\Theta\sum_{i=1}^N \lambda_i},
\end{align}
\end{small}
where (a) follows from Jensen's inequality.
The right hand side of~\eqref{eq:singleVSbatch_1} suggests that batching has no effect on the algorithm's regret.

In the second case, for $t>1$, $\tilde{g}_{t,i} = g_{t-1,i}=n_{t,i}(\tau)\sim \textrm{Poisson}(\lambda_i\tau)$, where $n_{t,i}(\tau)$ is the number of arrivals within the interval $[(t-2)\tau,(t-1) \tau]$.
The initial prediction is given by $\tilde{g}_{1,i} = \frac{n_i(\tau_0)}{\tau_0}$, where $\tau_0$ is a first warm-up interval. Looking at the expectation of $(g_{t,i}-\tilde{g}_{t,i})^2$, we have
\begin{gather}
    \mathbb{E}[(g_{t,i}-\tilde{g}_{t,i})^2] = \notag\\ \mathbb{E}[(g_{t,i}-\tilde{g}_{t,i} - \mathbb{E}[g_{t,i}]+\mathbb{E}[g_{t,i}]-\mathbb{E}[\tilde{g}_{t,i}]+\mathbb{E}[\tilde{g}_{t,i}])^2] = \notag \\
    \textrm{Var}(g_{t,i})+\textrm{Var}(\tilde{g}_{t,i}) + (\mathbb{E}[g_{t,i}] - \mathbb{E}[\tilde{g}_{t,i}])^2 = \notag \\
    = \begin{cases}
    2\lambda_i\tau, & t > 1\\
    \lambda_i\tau + (\frac{\tau}{\tau_0})^2\lambda_i\tau_0, & t=1.
    \end{cases}
\end{gather}
Summing all the terms over $N$ and $\Theta/\tau$, we obtain
\begin{equation}\label{eq:expected1}
    \sum_{i=1}^N \sum_{t=1}^{\Theta/\tau} \mathbb{E}[(g_{t,i}-\tilde{g}_{t,i})^2] = \sum_{i=1}^N \frac{\Theta-\tau}{\tau} 2\lambda_i\tau + \lambda_i\tau +m^2\tau^2,
\end{equation}
where $m^2 \triangleq \frac{\lambda_i\tau_0}{\tau_0^2}$. 
Under these predictions, there is indeed an optimal timescale $\tau^*$ for batching, that is
$\tau^* = \min\{\frac{\tau_0}{2},\Theta\}$. Hence, in case of a good initial prediction (large $\tau_0$) we should select $\tau = \Theta$. Otherwise, in case of a less accurate initial prediction we should choose a smaller value $\tau= \frac{\tau_0}{2}$.

\section{Numerical Results}\label{sec:experiments}
\begin{figure}[t]
\captionsetup[subfigure]{aboveskip=-1pt,belowskip=-1pt}
        \begin{subfigure}{0.4\linewidth}
        \centering
        \includegraphics[scale=0.4]{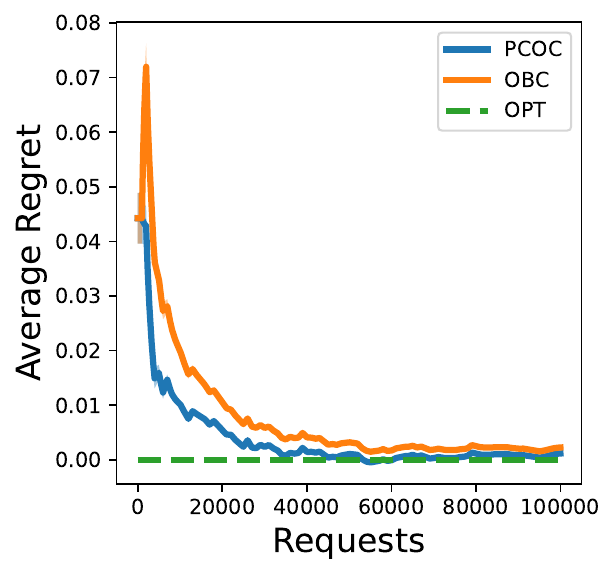}
        \caption{$\beta=0.8$, $k=50$}
        \label{fig:0.8-50}
        \end{subfigure}
        \hspace{10mm}
        \begin{subfigure}{0.4\linewidth}
        \centering
        \includegraphics[scale=0.4]{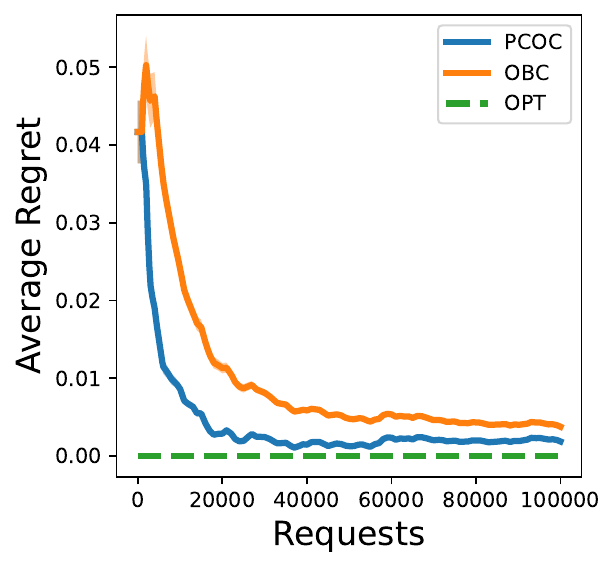}
        \caption{$\beta=1.5$, $k=50$}
        \label{fig:1.5-50}
        \end{subfigure}\\
        \begin{subfigure}{0.4\linewidth}
        \centering
        \includegraphics[scale=0.4]{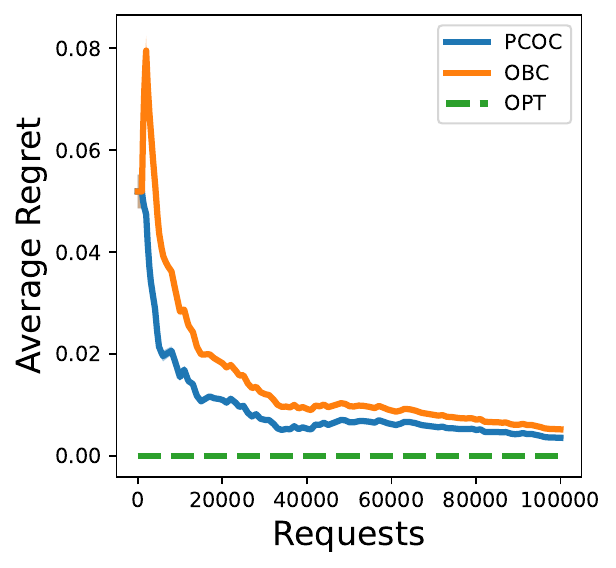}
        \caption{$\beta=0.8$, $k=100$}
        \label{fig:0.8-100}
        \end{subfigure}
        \hspace{10mm}
        \begin{subfigure}{0.4\linewidth}
        \centering
        \includegraphics[scale=0.4]{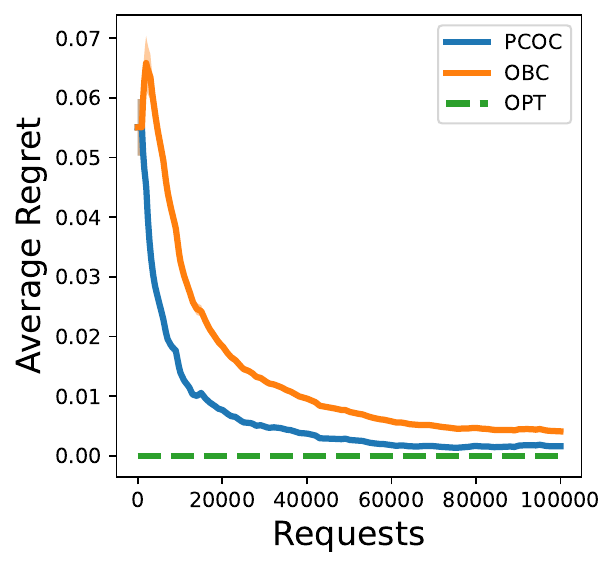}
        \caption{$\beta=1.5$, $k=100$}
        \label{fig:1.5-100}
        \end{subfigure}\\
        \begin{subfigure}{0.4\linewidth}
        \centering
        \includegraphics[scale=0.4]{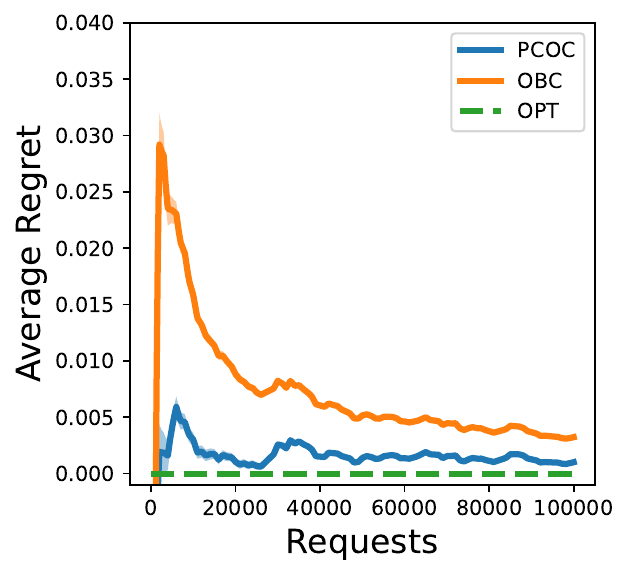}
        \caption{$\beta=0.8$, $k=600$}
        \label{fig:0.8-600}
        \end{subfigure}
        \hspace{10mm}
        \begin{subfigure}{0.4\linewidth}
        \centering
        \includegraphics[scale=0.4]{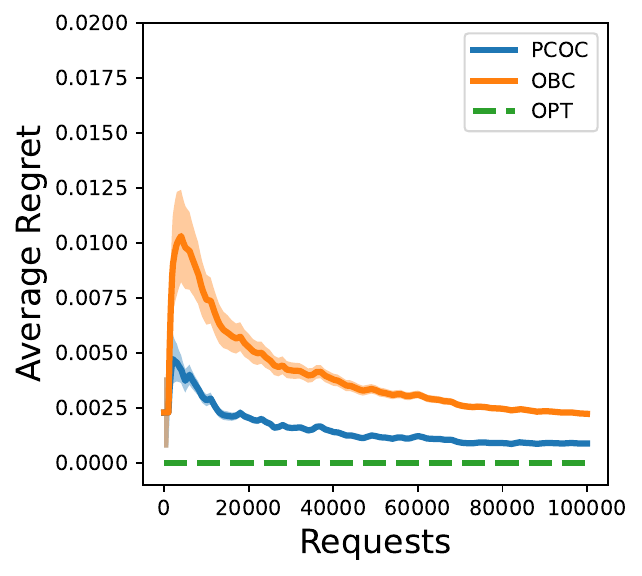}
        \caption{$\beta=1.5$, $k=600$}
        \label{fig:1.5-600}
        \end{subfigure}
        \caption{\ff{PCOC vs. OBC}}
        \label{fig:PCOCvsOBC}
\end{figure}
\begin{figure*}[t]
    \begin{subfigure}{0.32\linewidth}
  \centering
  \includegraphics[scale=0.5]{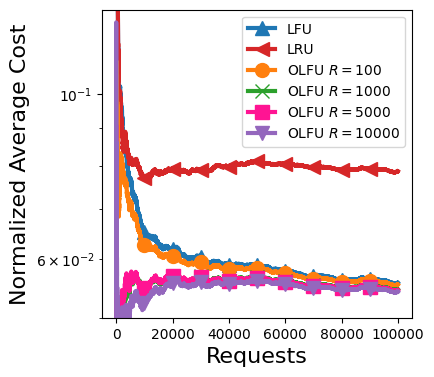}
\caption{$\pi = 1$}
\label{fig:p0_lfu}
\end{subfigure}
\begin{subfigure}{0.32\linewidth}
  \centering
  \includegraphics[scale=0.5]{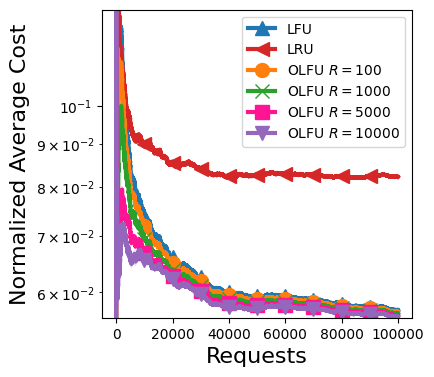}
  \caption{$\pi = 0.7$}
  \label{fig:p0.3_lfu}
\end{subfigure}
 \begin{subfigure}{0.32\linewidth}
      \centering  
    \includegraphics[scale=0.5]{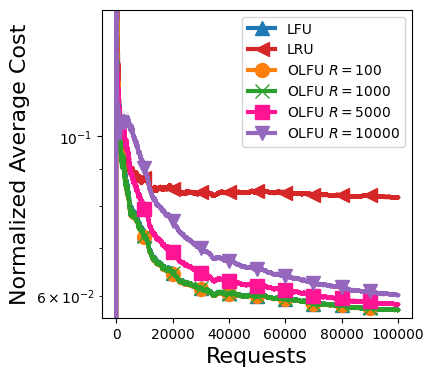}
    \caption{$\pi = 0.1$}
    \label{fig:p0.9_lfu}
\end{subfigure}
    \caption{\ff{Average Miss Ratio of OLFU vs. LFU}}
    \label{fig:OLFUvsLFU} 
    \vspace{-0.1in}
\end{figure*}
\begin{figure*}
    \begin{subfigure}{0.32\linewidth}
  \centering
  \includegraphics[scale=0.5]{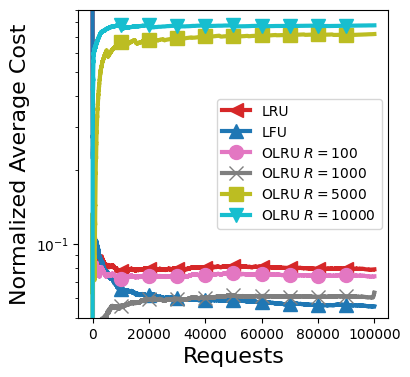}
\caption{$\pi = 1$}
\label{fig:p0_lru}
\end{subfigure}
\begin{subfigure}{0.32\linewidth}
  \centering
  \includegraphics[scale=0.5]{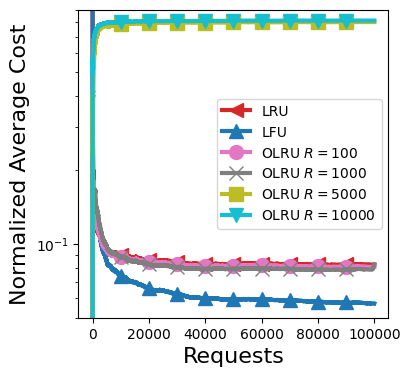}
  \caption{$\pi = 0.7$}
  \label{fig:p0.3_lru}
\end{subfigure}
 \begin{subfigure}{0.32\linewidth}
      \centering  
    \includegraphics[scale=0.5]{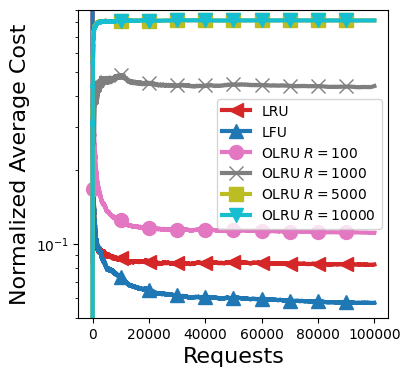}
    \caption{$\pi = 0.1$}
    \label{fig:p0.9_lru}
\end{subfigure}
    \caption{\ff{Average Miss Ratio of OLRU vs. LRU}}
    \label{fig:OLRUvsLRU} 
    \vspace{-0.1in}
\end{figure*}
\begin{figure*}
    \begin{subfigure}{0.32\linewidth}
  \centering
  \includegraphics[scale=0.5]{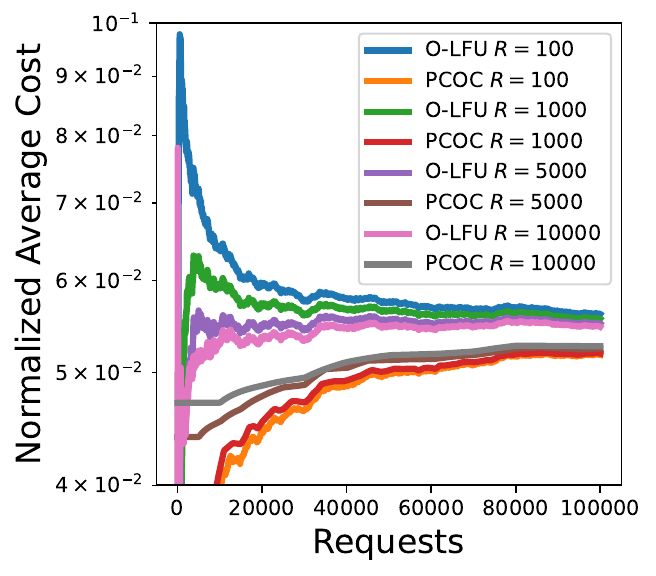}
\caption{$\pi = 1$}
\label{fig:p0}
\end{subfigure}
\begin{subfigure}{0.32\linewidth}
  \centering
  \includegraphics[scale=0.5]{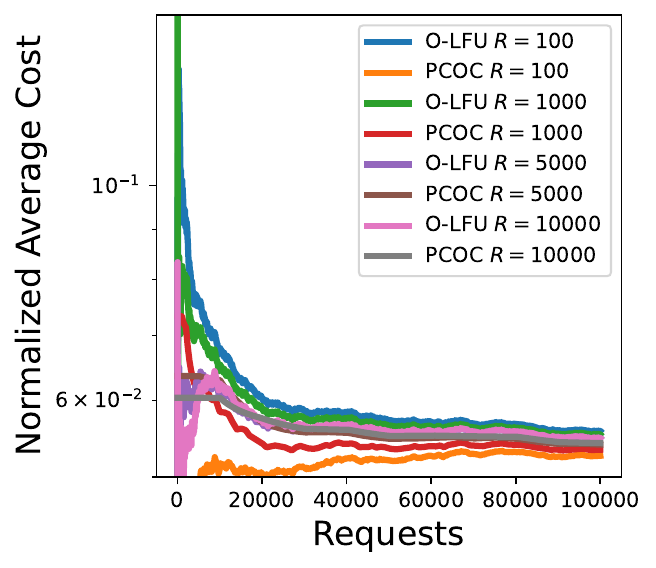}
  \caption{$\pi = 0.7$}
  \label{fig:p0.3}
\end{subfigure}
 \begin{subfigure}{0.32\linewidth}
      \centering  
    \includegraphics[scale=0.5]{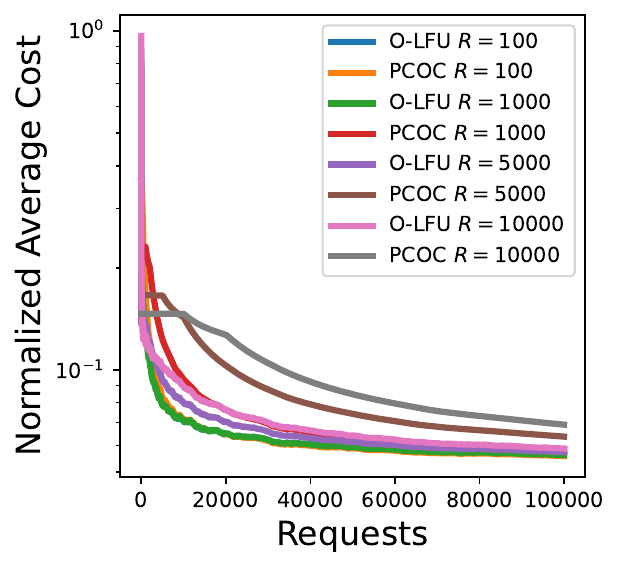}
    \caption{$\pi = 0.1$}
    \label{fig:p0.9}
\end{subfigure}
    \caption{\ff{Average Miss Ratio of PCOC vs. OLFU}}
    \label{fig:PCOCvsOLFU} 
    \vspace{-0.1in}
\end{figure*}
\begin{figure*}[t]
\captionsetup[subfigure]{aboveskip=-1pt,belowskip=-1pt}
\begin{subfigure}{0.5\linewidth}
  \centering
  \includegraphics[scale=0.6]{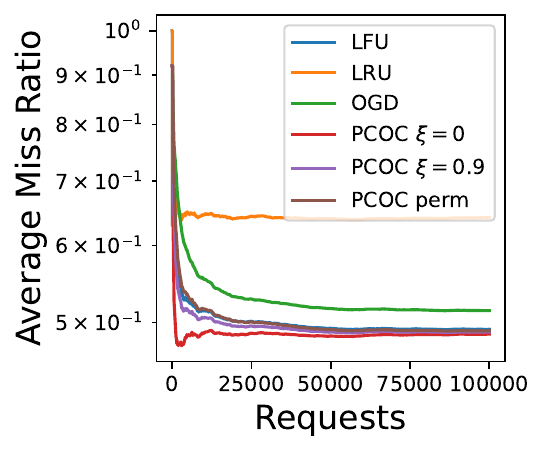}
  \caption{$\beta=0.9$, $k=50$}
  \label{fig:PCOCvsClassicsA}
\end{subfigure} 
\begin{subfigure}{0.5\linewidth}
  \centering
  \includegraphics[scale=0.6]{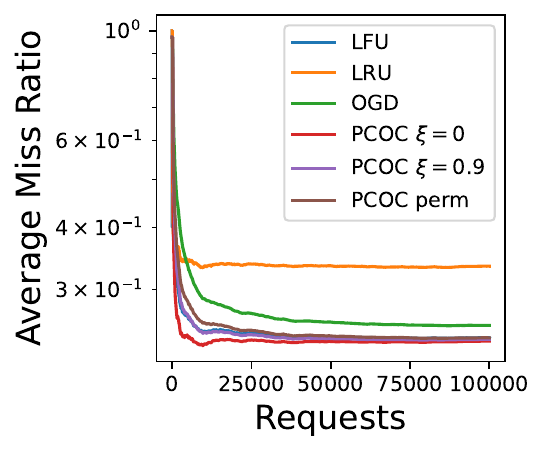}
  \caption{$\beta=1.2$, $k=50$}
  \label{fig:PCOCvsClassicsB}
\end{subfigure}
\caption{PCOC vs. Classic Policies}
\label{fig:PCOCvsClassics}
\end{figure*}
\begin{figure*}[t]
\captionsetup[subfigure]{aboveskip=-1pt,belowskip=-1pt}
\begin{subfigure}{0.32\linewidth}
  \centering
  \includegraphics[scale=0.5]{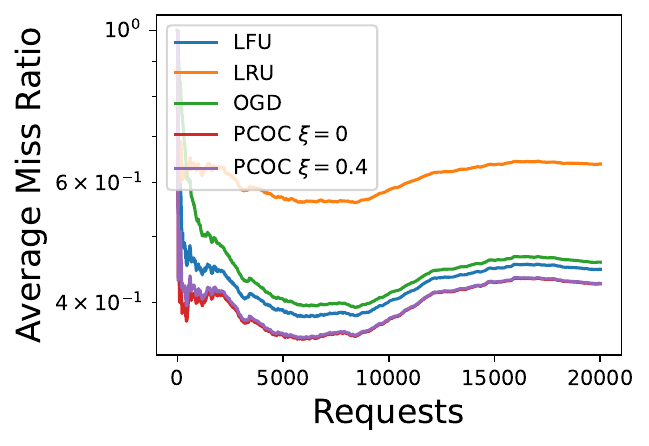}
  \caption{$R=10$, $k=10$}
  \label{fig:RealAndTimeA}
\end{subfigure}
\begin{subfigure}{0.32\linewidth}
  \centering
  \includegraphics[scale=0.5]{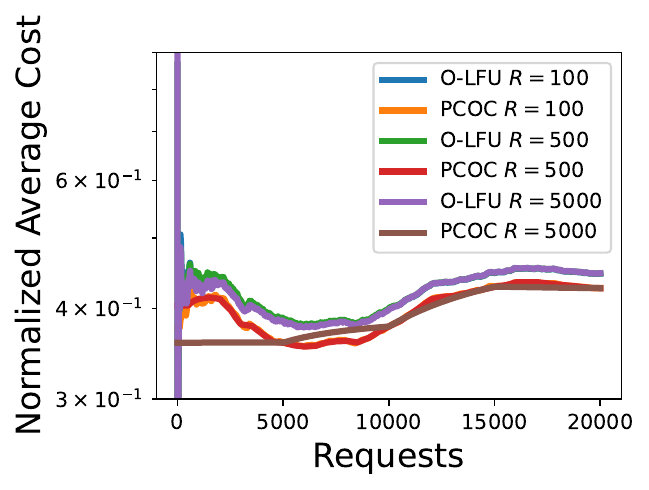}
  \caption{$k=10, \pi = 0.7$}
  \label{fig:RealAndTimeC}
\end{subfigure}%
\begin{subfigure}{0.32\linewidth}
  \centering
  \includegraphics[scale=0.5]{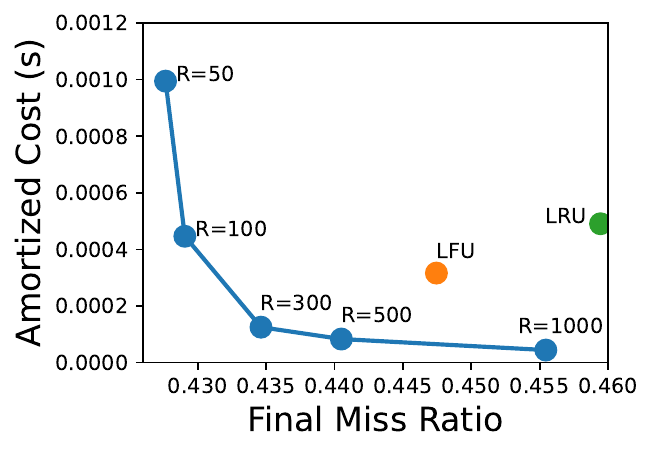}
  \caption{$k=10$, $\xi=0.4$}
  \label{fig:RealAndTimeB}
\end{subfigure}
\caption{Akamai Trace}
\label{fig:RealAndTime}
\end{figure*}
\subsection{Experimental Settings}
\subsubsection{Datasets}
\ff{We evaluated the presented approaches on both synthetic and real traces. For the synthetic case, we generated stationary synthetic traces where individual file requests are generated i.i.d. according to a Zipf distribution with parameter $\beta \in \{0.8,1.2,1.5\}$ from a catalog of $N=1000$ files. 
We evaluate the studied solutions against state-of-the-art algorithms over a horizon of $I=10^5$ requests. 
\textit{Batched} algorithms have a constant batch size, i.e., $R_t = R$ with $R \in \{100,1000,2000,5000,10000\}$ for synthetic traces, and $R\in\{10,50,100,300,1000\}$ for the real trace. 
The cache size $k$ varies in $\{10,50,100,600\}$. 
The real trace counts  $2\cdot10^4$ requests for the $N=10^3$ most popular files as measured at a given server in Akamai CDN provider~\cite{neglia2017access}. In all the experiments we set $w_i = 1, \forall i \in \mathcal{N}$, the cost in~\eqref{eq:cost} corresponds then to the total number of misses. In Figures~\ref{fig:PCOCvsOBC},\ref{fig:OLFUvsLFU},\ref{fig:OLRUvsLRU},\ref{fig:PCOCvsOLFU} and~\ref{fig:RealAndTimeC}, given a vector of requests over the time horizon $T$, we report the average over 30 different runs for predictions and we also plot the 0.95-confidence interval of the normalized average cost and the average regret.} 
\subsubsection{Predictions}
\ff{For the optimistic algorithms' evaluation we considered three types of predictions: \\ 
\textit{Type 1}: the first ones are generated according to $\tilde{\mathbf{g}}_t = (1-\xi)\mathbf{g}_t + \xi \frac{R}{N}$, with $\xi \in [0,1]$;\\
\textit{Type 2}: the second ones are generated as random permutations of the correct gradients;\\ 
\textit{Type 3}: the third case is the same described in~\cite{mhaisen2022optimistic}, where each prediction is assumed to be correct with a probability $\pi$.\\
The first type interpolates between perfect predictions (for $\xi=0$) and a situation where all files appear equally popular (for $\xi=1$). In the second type, files' future popularities are arbitrarily ranked. In the latter case, given the original vector of requests $\mathbf{r}_t$, the prediction vector $\tilde{\mathbf{g}}_t$ is generated by requesting the original files in $\mathbf{r}_t$ with probability $\pi$ and any other random file from the catalog with probability $1-\pi$. }
\subsubsection{Metrics}
We evaluate all the algorithms according to three metrics:\\
i) the \textit{Average Miss Ratio}, i.e., the total cost over the first $t$ iterations, normalized by~$R t$;\\
ii) the \textit{Time Average Regret} over the first $t$ iterations;\\ 
iii) the \textit{Amortized Cost}, i.e., the average computational time per request.
\subsubsection{Online Algorithms}
\ff{We compare OBC and PCOC presented in Section~\ref{sec:optimistic} against classical online algorithms such as LFU, LRU, and OGD~\cite{paschos2019learning}. Furthermore, we designed and implemented optimistic version of LFU and LRU.}

\ff{\textbf{Optimistic Least Frequently Used (OLFU)}. The algorithm takes into account predictions for the next requests but updates the cache state at each new requests according to the LFU eviction policy. At the beginning of each batch of requests, OLFU increases the frequency of each file within the predictions for the next batch of $R$ requests. In the face of a new request, the algorithm i) updates the cache state using LFU with the updated frequencies; ii) checks if the file request was in the predicted batch: if it was not, OLFU increases the frequency for that file and decreases the frequency of a random file from the catalog different from the requested one. At the end of each batch the frequencies of OLFU and the ones computed by a classic LFU policy are equal.}

\ff{\textbf{Optimisitc Least Recently Used (OLRU)}. This policy considers the predictions for the next $R$ requests and consider the files within the batch as the most recently requested. For each file $i \in \mathcal{N}$, the algorithm keeps a counter, namely \textit{last-time-requested}, indicating the last time file $i$ has been requested. In particular, given a batch of predicted requests, OLRU sets the \textit{last-time-requested} counter of all those predicted files to the current time. In the face of a new request, the algorithm updates the cache using LRU, i.e., evicting the least recently used file from the cache according to the counters updated through the predictions.}

\subsection{Results}
\ff{First of all we compare the optimistic versions of LFU and LRU with respect to their classical versions. Afterwards, we focus on the Follow-The-Regularized-Leader-based algorithms evaluating their performance in terms of average regret. Consequently, we compare PCOC with respect to OLFU and classical policies. Finally, we evaluate the optimistic versions of the presented algorithms on the Akamai trace showing also the trade-off between the final missing-ratio and the amortized cost varying the batch size.}

\ff{\textbf{OLFU vs. LFU}. Figure~\ref{fig:OLFUvsLFU} compares OLFU against LFU for different batch sizes and for different levels of the predictions' accuracy with predictions of \textit{Type 3}. We can observe that the batch size plays an important role in the performance of OLFU as the predictions become worse. Indeed, in case of perfect predictions (Figure~\ref{fig:p0_lfu}), the versions of OLFU with the highest batch sizes reach a better miss-ratio with respect to LFU since, as the batch size increases, there is more accurate information about the next requests. On the other hand, with very inaccurate predictions (Figure~\ref{fig:p0.9_lfu}), the higher is the batch size and the worse is the missing-ratio, given the incorrect information brought by the perturbed predictions.}

\ff{\textbf{OLRU vs. LRU}. In contrast with OLFU, as highlighted in Figure~\ref{fig:OLRUvsLRU}, the optimistic version of LRU performs better for small batch sizes as the predictions' accuracy deteriorates. Indeed, the bigger is the batch and the fewer will be the number of cache updates. In this manner, the counters of all the files within the batch will be updated less frequently resulting to be stale. Beyond such a staleness, the performance of the policy deteriorates as the inaccuracy of the predictions increases.}

\textbf{PCOC vs. OBC}. \ff{We compare the two algorithms for different capacities, i.e., $k \in \{50,100,600\}$ and different exponents of the Zipf distribution, i.e., $\beta \in \{0.8,1.5\}$ with $R=1000$ with predictions of \textit{Type 3}. As showed in Figure~\ref{fig:PCOCvsOBC} the difference between the two algorithms becomes significant as the values of $\alpha$ and $\beta$ increase. This confirms the results of Figure~\ref{fig:aVSb} where the difference between the two regrets becomes more evident for higher values of the cache size and the Zipf's exponent. 
In particular when $k=600$, i.e., the cache can store at
least half of the catalog, PCOC clearly outperforms OBC for all the values of $\beta$.}

\ff{ \textbf{PCOC vs. OLFU}. Figure~\ref{fig:PCOCvsOLFU} reports on the comparison between PCOC and OLF for different batch sizes and levels of accuracy in predictions of \textit{Type 3}. For all the algorithms we set the initial cache state as $\mathbf{x}_0 := \arg\max_{x \in \mathcal{X}} \{\tilde{\mathbf{r}}_{1}^\top x\}$, i.e., we entirely store the files with the highest number of requests in the first predicted batch. We can observe that for high levels of accuracy in the predictions (Figure~\ref{fig:p0} and Figure~\ref{fig:p0.3}) PCOC outperforms OLFU for all the different batches. When the predictions have very low accuracy ($\pi=0.1$) PCOC shows the same performance of OLFU for $R=100$, however it still remains competitive reaching the convergence even for higher values of $R$.}

\textbf{PCOC vs. Classic Policies}. Figures~\ref{fig:PCOCvsClassicsA} and \ref{fig:PCOCvsClassicsB} show the performance of PCOC against classical online algorithms in cases where $\beta=0.9$, and $\beta=1.2$, with $R=100$ with predicitons of \textit{Type 1} and \textit{Type 2}. We can notice the benefit of including predictions in the decision process looking at the lower miss ratio of PCOC against LFU. PCOC outperforms LFU even for a noisy factor $\xi$ as large as 0.9 and it is still competitive with LFU when predictions are randomly scrambled.
This confirms the advantage of the optimistic nature of such algorithms.

\textbf{Akamai Trace}. Figure~\ref{fig:RealAndTime} shows the performance of PCOC on the Akamai trace for $k=10$ with predictions of \textit{Type 1}. Figure~\ref{fig:RealAndTimeA} compares PCOC against OGD, LFU and LRU. The latter two policies take a decision at each file request, whilst PCOC and OGD updates the cache every $R=10$ requests. Nevertheless, PCOC outperforms the classic policies. Furthermore, even in a non-stationary case, the predictions can help in reducing the miss ratio. \ff{Figure~\ref{fig:RealAndTimeC} shows the comparison between PCOC and OLFU for different batch sizes and with predictions of \textit{Type 3} with $\pi = 0.7$. We can notice how the difference between the two policies becomes more evident in case of real trace even for higher batch sizes for PCOC.} Finally, in Figure~\ref{fig:RealAndTimeB}, we compare different versions of PCOC that updates the local cache every $R \in \{50,100,300,500,1000\}$ requests. 
The amortized cost vanishes as the value of $R$ increases (since the number of projections performed in the optimization process diminishes) at the cost of higher miss ratio. However, this confirms the applicability of such a batched method with less frequent updates since both the final miss ratio and the time complexity reached by PCOC with $R=300$ and $R=500$ are better than the performance achieved by the most used policies in practice such as LFU and LRU.

\section{Conclusions}\label{sec:conclusions}
We presented online optimistic caching algorithms that enjoy sublinear regret in case of batched requests. First we studied the conditions where PCOC results to have a better regret with respect to OBC. Secondly, we showed that the per-component based solution (PCOC) outperforms classic caching policies and their optimistic versions in different conditions. Finally, we showed that, over a real trace, a batched approach presents better performance in terms of final miss ratio and amortized cost compared to classical caching policies. 

\bibliographystyle{elsarticle-num} 
\bibliography{main}

\begin{thebibliography}{10}
\expandafter\ifx\csname url\endcsname\relax
  \def\url#1{\texttt{#1}}\fi
\expandafter\ifx\csname urlprefix\endcsname\relax\def\urlprefix{URL }\fi
\expandafter\ifx\csname href\endcsname\relax
  \def\href#1#2{#2} \def\path#1{#1}\fi

\bibitem{icc2023}
F.~Faticanti, G.~Neglia, {Optimistic Online Caching for Batched Requests}, in:
  IEEE ICC, 2023, pp. 1--6.

\bibitem{carra2020elastic}
D.~Carra, G.~Neglia, P.~Michiardi, Elastic provisioning of cloud caches: A
  cost-aware ttl approach, IEEE/ACM Transactions on Networking 28~(3) (2020)
  1283--1296.

\bibitem{carlsson2018worst}
N.~Carlsson, D.~Eager, Worst-case bounds and optimized cache on mth request
  cache insertion policies under elastic conditions, Performance Evaluation 127
  (2018) 70--92.

\bibitem{paschos2019learning}
G.~S. Paschos, A.~Destounis, L.~Vigneri, G.~Iosifidis, Learning to cache with
  no regrets, in: IEEE INFOCOM 2019-IEEE Conference on Computer Communications,
  IEEE, 2019, pp. 235--243.

\bibitem{bhattacharjee2020fundamental}
R.~Bhattacharjee, S.~Banerjee, A.~Sinha, Fundamental limits on the regret of
  online network-caching, Proceedings of the ACM on Measurement and Analysis of
  Computing Systems 4~(2) (2020) 1--31.

\bibitem{li2021online}
Y.~Li, T.~Si~Salem, G.~Neglia, S.~Ioannidis, Online caching networks with
  adversarial guarantees, Proceedings of the ACM on Measurement and Analysis of
  Computing Systems 5~(3) (2021) 1--39.

\bibitem{leconte2016placing}
M.~Leconte, G.~Paschos, L.~Gkatzikis, M.~Draief, S.~Vassilaras, S.~Chouvardas,
  Placing dynamic content in caches with small population, in: IEEE INFOCOM
  2016-The 35th Annual IEEE International Conference on Computer
  Communications, IEEE, 2016, pp. 1--9.

\bibitem{shalev2012online}
S.~Shalev-Shwartz, Online learning and online convex optimization, Foundations
  and Trends in Machine Learning 4~(2).

\bibitem{salem2021no}
T.~S. Salem, G.~Neglia, S.~Ioannidis, No-regret caching via online mirror
  descent, in: IEEE ICC, 2021, pp. 1--6.

\bibitem{gomez2015netflix}
C.~A. Gomez-Uribe, N.~Hunt, The netflix recommender system: Algorithms,
  business value, and innovation, ACM TMIS.

\bibitem{khanal2020systematic}
S.~S. Khanal, P.~Prasad, A.~Alsadoon, A.~Maag, A systematic review: machine
  learning based recommendation systems for e-learning, Education and
  Information Technologies 25 (2020) 2635--2664.

\bibitem{mohri2016accelerating}
M.~Mohri, S.~Yang, Accelerating online convex optimization via adaptive
  prediction, in: AISTAS, PMLR, 2016, pp. 848--856.

\bibitem{mhaisen2022online}
N.~Mhaisen, G.~Iosifidis, D.~Leith, Online caching with optimistic learning,
  in: 2022 IFIP Networking, IEEE, 2022, pp. 1--9.

\bibitem{rakhlin2013optimization}
S.~Rakhlin, K.~Sridharan, Optimization, learning, and games with predictable
  sequences, NeurIPS 26.

\bibitem{paschos2020cache}
G.~Paschos, G.~Iosifidis, G.~Caire, et~al., Cache optimization models and
  algorithms, Foundations and Trends{\textregistered} in Communications and
  Information Theory 16~(3--4) (2020) 156--345.

\bibitem{borst2010distributed}
S.~Borst, V.~Gupta, A.~Walid, Distributed caching algorithms for content
  distribution networks, in: 2010 Proceedings IEEE INFOCOM, IEEE, 2010, pp.
  1--9.

\bibitem{shanmugam2013femtocaching}
K.~Shanmugam, N.~Golrezaei, A.~G. Dimakis, A.~F. Molisch, G.~Caire,
  Femtocaching: Wireless content delivery through distributed caching helpers,
  IEEE Transactions on Information Theory 59~(12) (2013) 8402--8413.

\bibitem{poularakis2016exploiting}
K.~Poularakis, G.~Iosifidis, V.~Sourlas, L.~Tassiulas, Exploiting caching and
  multicast for 5g wireless networks, IEEE Transactions on Wireless
  Communications 15~(4) (2016) 2995--3007.

\bibitem{ioannidis2010distributed}
S.~Ioannidis, L.~Massoulie, A.~Chaintreau, Distributed caching over
  heterogeneous mobile networks, in: Proceedings of the ACM SIGMETRICS
  international conference on Measurement and modeling of computer systems,
  2010, pp. 311--322.

\bibitem{ioannidis2016adaptive}
S.~Ioannidis, E.~Yeh, Adaptive caching networks with optimality guarantees, ACM
  SIGMETRICS Performance Evaluation Review 44~(1) (2016) 113--124.

\bibitem{paria2021texttt}
D.~Paria, A.~Sinha, Leadcache: Regret-optimal caching in networks, Advances in
  Neural Information Processing Systems 34 (2021) 4435--4447.

\bibitem{paschos2020online}
G.~S. Paschos, A.~Destounis, G.~Iosifidis, Online convex optimization for
  caching networks, IEEE/ACM Transactions on Networking 28~(2) (2020) 625--638.

\bibitem{sleator1985amortized}
D.~D. Sleator, R.~E. Tarjan, Amortized efficiency of list update and paging
  rules, Communications of the ACM 28~(2) (1985) 202--208.

\bibitem{andrew2013tale}
L.~Andrew, S.~Barman, K.~Ligett, M.~Lin, A.~Meyerson, A.~Roytman, A.~Wierman, A
  tale of two metrics: Simultaneous bounds on competitiveness and regret, in:
  Conference on Learning Theory, PMLR, 2013, pp. 741--763.

\bibitem{zinkevich2003online}
M.~Zinkevich, Online convex programming and generalized infinitesimal gradient
  ascent, in: ICML 2003, 2003, pp. 928--936.

\bibitem{mhaisen2022optimistic}
N.~Mhaisen, A.~Sinha, G.~Paschos, G.~Iosifidis, Optimistic no-regret algorithms
  for discrete caching, Proceedings of the ACM on Measurement and Analysis of
  Computing Systems 6~(3) (2022) 1--28.

\bibitem{chen2018timely}
K.~Chen, L.~Huang, Timely-throughput optimal scheduling with prediction,
  IEEE/ACM Transactions on Networking 26~(6) (2018) 2457--2470.

\bibitem{huang2021online}
X.~Huang, S.~Bian, X.~Gao, W.~Wu, Z.~Shao, Y.~Yang, J.~C. Lui, Online vnf
  chaining and predictive scheduling: Optimality and trade-offs, IEEE/ACM
  Transactions on Networking 29~(4) (2021) 1867--1880.

\bibitem{salem2023no}
T.~S. Salem, G.~Neglia, S.~Ioannidis, No-regret caching via online mirror
  descent, arXiv preprint arXiv:2101.12588.

\bibitem{wang2015projection}
W.~Wang, C.~Lu, Projection onto the capped simplex, arXiv preprint
  arXiv:1503.01002.

\bibitem{golrezaei2013femtocaching}
N.~Golrezaei, A.~F. Molisch, A.~G. Dimakis, G.~Caire, Femtocaching and
  device-to-device collaboration: A new architecture for wireless video
  distribution, IEEE Communications Magazine 51~(4) (2013) 142--149.

\bibitem{hazan2016introduction}
E.~Hazan, Introduction to online convex optimization, Foundations and
  Trends{\textregistered} in Optimization 2~(3-4) (2016) 157--325.

\bibitem{mcmahan2017survey}
H.~B. McMahan, A survey of algorithms and analysis for adaptive online
  learning, The Journal of Machine Learning Research.

\bibitem{littlestone1994weighted}
N.~Littlestone, M.~K. Warmuth, The weighted majority algorithm, Information and
  computation 108~(2) (1994) 212--261.

\bibitem{auer2002adaptive}
P.~Auer, N.~Cesa-Bianchi, C.~Gentile, Adaptive and self-confident on-line
  learning algorithms, Journal of Computer and System Sciences 64~(1) (2002)
  48--75.

\bibitem{irm-fagin-1977}
R.~Fagin, Asymptotic miss ratios over independent references, Journal of
  Computer and System Sciences 14~(2) (1977) 222--250.

\bibitem{neglia2017access}
G.~Neglia, D.~Carra, M.~Feng, V.~Janardhan, P.~Michiardi, D.~Tsigkari,
  Access-time-aware cache algorithms, ACM Transactions on Modeling and
  Performance Evaluation of Computing Systems (TOMPECS) 2~(4) (2017) 1--29.

\end{thebibliography}

\end{document}